%% file: root.tex
\documentclass[letterpaper, 10pt, conference]{ieeeconf}  % Comment this line out if you need a4paper

\IEEEoverridecommandlockouts                              % This command is only needed if 
\usepackage[utf8]{inputenc}
\usepackage{authblk}
\usepackage{setspace}
\usepackage{xspace}
\usepackage{graphicx}
\graphicspath{ {./figures/} }
\usepackage{subcaption}
\usepackage{amsmath,amssymb}
\usepackage{mathtools}
\usepackage{booktabs}
\usepackage{tabularx}
\usepackage{multirow}
\usepackage{url}
\usepackage[]{todonotes}
\usepackage[acronym]{glossaries}
\input{acronyms.tex}
\usepackage[normalem]{ulem}

\input{ancillary/commands.tex}

\newtheorem{prop}{Proposition}

\usepackage[style=ieee, 
citestyle=numeric-comp, maxalphanames=1,maxcitenames=1,
sorting=none,backend=bibtex]{biblatex}
\title{Inner Approximations of Stochastic Programs for Data-driven Stochastic Barrier Function Design}
\author[1*$\dag$]{Frederik Baymler Mathiesen}
\author[2$\dag$]{Licio Romao}
\author[3]{Simeon C. Calvert}
\author[2]{Alessandro Abate}
\author[1]{Luca Laurenti}

\affil[1]{Delft Center for Systems and Control, TU Delft.}
\affil[2]{Department of Computer Science, University of Oxford.}
\affil[3]{Department of Transport \& Planning, TU Delft.}
\affil[*]{Corresponding author. Email: \texttt{frederik@baymler.com}.}
\affil[$\dag$]{These authors contributed equally to this work.}

\begin{document}

%%%%%% Date %%%%%%
% Date is optional
\date{}

\maketitle

%%%%%% Abstract %%%%%%
\begin{abstract}

This paper proposes a new framework to compute finite-horizon safety guarantees for discrete-time piece-wise affine systems with stochastic noise of unknown distributions.
The approach is based on a novel approach to synthesise a stochastic barrier function (SBF) from noisy data and rely on the scenario optimization theory. In particular, we show that the stochastic program to synthesize a SBF can be relaxed into a chance-constrained optimisation problem on which scenario approach theory applies. We further show that the resulting program can be reduced to a linear programming problem, thus guaranteeing efficiency. In contrast to existing approaches, this method is data efficient as it only requires the number of data to be proportional to the logarithm in the negative inverse of the confidence level and is computationally efficient due to its reduction to linear programming.
The efficacy of the method is empirically evaluated on various verification benchmarks.
Experiments show a significant improvement with respect to state-of-the-art, obtaining tighter certificates with a confidence that is several orders of magnitude higher.

\end{abstract}

\section{Introduction}
\label{sec:introduction}

\input{sections/introduction.tex}

\paragraph{Related works}
\input{sections/related_works.tex}

\section{Preliminaries}
\label{sec:prelim}

\input{sections/preliminaries.tex}

\section{Problem statement}
\label{sec:problem_statement}

\input{sections/problem_statement.tex}

\section{Stochastic barrier function (SBF)}
\label{sec:stochastic_barrier_functions}

\input{sections/stochastic_barrier_functions.tex}

\section{Data-driven stochastic barrier function design}
\label{sec:data_driven_sbf_design}

\input{sections/data_driven_sbf_design.tex}

\section{Experiments}
\label{sec:experiments}

\input{sections/experiments.tex}

\section{Conclusions}
\label{sec:conclusion}

\input{sections/conclusion.tex}

\section{Technical proofs}
\label{sec:technical_proofs}

\input{sections/technical_proofs.tex}

\printbibliography

\end{document}

%% file: acronyms.tex
\newacronym{pwa}{PWA}{piece-wise affine}
\newacronym{lp}{LP}{Linear Programming}
\newacronym{nndm}{NNDM}{Neural Network Dynamical Model}
\newacronym{sbf}{SBF}{Stochastic Barrier Function}
\newacronym{mbr}{MBR}{Minimum Bounding Rectangle}
\newacronym{saa}{SAA}{Sample Average Approximation}
\newacronym{ltl}{LTL}{Linear Temporal Logic}
\newacronym{pac}{PAC}{Probably Approximately Correct}
\newacronym{iid}{iid.}{independent and identically distributed}
\newacronym{sos}{SoS}{sum-of-squares}
\newacronym{lbp}{LBP}{Linear Bound Propagation}

%% file: ancillary/commands.tex
\newtheorem{defi}{Definition}

\newtheorem{assump}{Assumption}

\newtheorem{theorem}{Theorem}
\newtheorem{corollary}{Corollary}
\newtheorem{remark}{Remark}

\newcommand{\uncertaintyspace}{\ensuremath{\Omega}}
\newcommand{\uncertaintyelement}{\ensuremath{\omega}}
\newcommand{\sigmaalgebra}{\ensuremath{\mathcal{F}}}
\newcommand{\probmeas}{\ensuremath{\mathbb{P}}}
\newcommand{\noise}{\ensuremath{\eta}}

\newcommand{\sampleset}{{\ensuremath D}}

\newcommand{\barrierl}[0]{{\ensuremath u}}
\newcommand{\barrierc}[0]{{\ensuremath v}}
\newcommand{\cmartingale}[0]{{\ensuremath c}}
\newcommand{\buffervar}[0]{{\ensuremath \nu}}
\newcommand{\region}[0]{{\ensuremath P}}
\newcommand{\regionset}[0]{{\ensuremath \mathcal{P}}}
\newcommand{\numberregions}{\ensuremath{\ell}}

\newcommand{\safeset}[0]{{\ensuremath \mathcal{S}}}
\newcommand{\unsafeset}[0]{{\ensuremath \mathcal{U}}}
\newcommand{\initialset}[0]{{\ensuremath X_0}}

\DeclareMathOperator*{\minimise}{\min}

\DeclareMathOperator{\epi}{epi}

\DeclareMathOperator{\subjectto}{s.t.}
\newcommand{\dualvariable}{{\ensuremath \lambda}}

\newcommand{\polyhedron}{\ensuremath{P}}

%% file: sections/introduction.tex
The behavior of modern autonomous systems are often uncertain, due to e.g., sensor noise or unknown dynamics, and are commonly employed in safety-critical applications, such as automated driving \cite{Shalev-Shwartz2017} or robotics \cite{Livingston2012}. These applications require formal guarantees of safety in order for the system to be deployed in real-life. Consequently, computing such guarantees for stochastic systems represents an important, but non-trivial, area of research \cite{APLS:08}. 
Existing approaches to address this problem either rely on abstractions, where the original system is \emph{abstracted} into a finite state model, generally a variant of a Markov chain \cite{Cauchi2019}, or leverage the concept of \glspl{sbf} \cite{PJP:07}. \glspl{sbf} are Lyapunov-like functions that can be employed to bound the probability that a dynamical system will remain safe for a given time horizon, without the need to explicitly evolve the system over time.  
A common assumption for the vast majority of the existing approaches is that the distribution of the system is known, and often either Gaussian or of bounded support \cite{PJP:07, Santoyo2021}.
Unfortunately, in practice, the noise characteristics of the system are generally not known \cite{Gracia2022,rahimian2019distributionally}. This leads to the main question of this paper: \emph{how can we compute formal certificates of safety for stochastic systems with unknown noise distribution?}

This paper focuses on guaranteeing safety for stochastic \gls{pwa} systems. 
In particular, a data-driven framework for the design of \glspl{sbf} for stochastic \gls{pwa} systems with unknown noise distribution is presented. By relying on tools from probability theory and convex optimisation, we show that the problem of synthesising a \gls{sbf} for this class of systems can be reformulated as a chance-constrained optimisation problem \cite{Shapiro2021-zc}. This reformulation allows employing the scenario approach theory to devise a data-driven framework to synthesize \glspl{sbf} with high confidence. The resulting approach is data-efficient, as it only requires the amount of data to be logarithmic in the negative inverse of the confidence, and is scalable, as, in the case of \gls{pwa} \glspl{sbf}, it reduces to the solution of a \gls{lp} problem.
We experimentally evaluate the performance of the method on various systems including a model of a vehicle in windy conditions. Our analysis illustrates how our approach outperforms state-of-the-art comparable methods both in terms of tightness of bounds and amount of data required to achieve the desired confidence. 
In summary, the main contributions  are:
\begin{itemize}
    \item A data-driven method based on the scenario approach to design \acrlong{pwa} \acrlongpl{sbf}.
    \item A novel inner chance-constrained approximation to stochastic programming.
    \item Empirical studies that illustrates the performance of the proposed method compared to state-of-the-art in terms of both certified safety probability and confidence.
\end{itemize}

The structure of the paper is as follows: Section \ref{sec:prelim} reviews convex and scenario optimisation, which are used extensively throughout the paper. Section \ref{sec:problem_statement} describes the safety certification problem and Section \ref{sec:stochastic_barrier_functions} how \glspl{sbf} formally can guarantee safety. In Section \ref{sec:data_driven_sbf_design} are the main results of this paper; namely the inner approximation to stochastic programming and data-driven \gls{sbf} design. Empirical studies are reported in Section \ref{sec:experiments}.

%% file: sections/related_works.tex
\Glspl{sbf} were first proposed in \cite{kushner1967stochastic} to study the probability that a stochastic system exits a given set in a finite time using super-martingale theory. Since then, various works have employed \glspl{sbf} to study non-linear stochastic systems with approaches including \gls{sos} optimisation  \cite{PJP:07,Santoyo2021,pushpak2018,Salamati2023,abate2023quantitative} and relaxations to convex programming \cite{mazouzsafety2022, Mathiesen2013}.
However, all these methods assume that the model of the system is fully known. A recent set of works have started  to study data-driven approaches to design \glspl{sbf} for stochastic systems with (partially) unknown dynamics, which can be employed to obtain guarantees of safety with a confidence \cite{SALAMATI20217, Salamati2023}. These approaches replace the stochastic program for synthesising \glspl{sbf} with a \gls{saa}-based program, meaning that the expectation is replaced by the sample average with a probabilistic guarantee of satisfaction of the original expectation constraint through concentration inequalities. However, these methods require an amount of data that is proportional to the negative inverse of the confidence.
In contrast, our approach requires a number of data that is logarithmic in the negative inverse of the confidence.

Data-driven verification of stochastic systems is a relatively new area to address the problem of verifying (partially) unknown systems \cite{cAP11,Badings_2022_1,Badings2022, Abolfazl2023, SALAMATI20217, Salamati2023}.
To compute formal guarantees for non-linear systems, apart from the SAA approach described in the previous paragraph, existing literature focuses either on the scenario approach \cite{Badings_2022_1,Badings2022, Abolfazl2023}, on Gaussian processes \cite{Jackson2021, HASHIMOTO2022110646}, or on distributionally robust approaches \cite{Gracia2022}.
In particular, in \cite{Badings_2022_1,Badings2022,Abolfazl2023} the authors rely on the data efficiency of scenario approach theory to build abstractions of the original system with high confidence of correctness, while in \cite{Jackson2021,HASHIMOTO2022110646} error bounds on performing Gaussian Process regression are employed to again build  abstractions that are employed to perform probabilistic model checking of the unknown system. However, all these methods are abstraction-based. Consequently, they suffer from the scalability issues inherent with abstraction-based frameworks. In this paper, our approach will combine the data-efficiency of the scenario approach with the flexibility of \glspl{sbf}.

\subsection{Notation} The set of real, non-negative real, and natural numbers are denoted with $\mathbb{R}$, $\mathbb{R}_{\geq 0}$, and $\mathbb{N}$ respectively. Vectors in the Euclidean space will be denoted by the letter $x \in \mathbb{R}^n$ and random variables in $\mathbb{R}^n$ will be denoted with bold font $\mathbf{x}$. Subscripts will be used to denote a collection of elements, i.e., $x_1, \ldots, x_m$ denote different vectors in the same space. A subset $X$ of $\mathbb{R}^n$ is convex if $\lambda x_1 + (1-\lambda) x_2 \in X$, for all $x_1, x_2 \in X$ and $\lambda \in [0,1]$. A polyhedron $\polyhedron \subseteq \mathbb{R}^n$ is a convex set defined as $\polyhedron = \{x \in \mathbb{R}^n : Hx \leq h \},$ where the matrix $H \in \mathbb{R}^{m \times n}$ and the vector $h \in \mathbb{R}^m$ are given, and the inequality is interpreted element-wise. This form is called a half-space representation. A function $f:\mathbb{R}^n \mapsto \mathbb{R}$ is convex if and only if its epigraph $\epi(f)$, defined as $\epi(f) = \{(x,t) \in \mathbb{R}^{n+1}: f(x) \leq t\}$, is a convex set of $\mathbb{R}^{n+1}$. Optimisation variables will be denoted by the letter $z$ to distinguish it from the state-space variable $x$.

%% file: sections/preliminaries.tex
In this section, we review some concepts used extensively throughout the paper.

\subsection{Robust linear programming}
\label{subsec:prelim_robust_lp}
Robust linear programming (\gls{lp}) \cite{BV:04} forms a backbone in this paper, hence we will reiterate its definition and crucial results. 
Consider the following robust \gls{lp} problem for polyhedron $\polyhedron \subset \mathbb{R}^n$
\begin{equation}
\begin{aligned}
\minimise_{z} & \quad s^\top z \\ 
\subjectto & \quad a(z)^\top x \leq b(z), \quad \text{ for all }x\in \polyhedron
\end{aligned}
\label{eq:robust_LP}
\end{equation}
where $z \in \mathbb{R}^d$ is the decision variable, $s \in \mathbb{R}^d$ is the cost vector, and $a : \mathbb{R}^d \to \mathbb{R}^n, b : \mathbb{R}^d \to \mathbb{R}$ are functions affine in $z$.
The following result guarantees that Problem \ref{eq:robust_LP} can be reformulated as a \gls{lp} problem in a lifted space.
\begin{prop}[Strong duality of robust \gls{lp} \cite{BV:04}] Consider the robust \gls{lp} problem in Problem \eqref{eq:robust_LP} and the following optimisation problem
    \begin{equation}
    \begin{aligned}
        \minimise_{z,\lambda} & \quad s^\top z \\
        \subjectto & \quad h^\top \lambda \leq b(z) \\ 
        & \quad H^\top \lambda = a(z), \quad \lambda \geq 0.
    \end{aligned}
    \label{eq:dual_robust_LP_formulation}
    \end{equation}
    where $(H, h)$ is the half-space representation of $\polyhedron$.
    Let sets
    \begin{align*}
        \mathcal{Z} &= \{ z \in \mathbb{R}^d : \sup_{x\in \polyhedron} a(z)^\top x \leq b(z) \}, \\
        \mathcal{Z}' &= \{ z \in \mathbb{R}^d : \exists \lambda \in \mathbb{R}^m_{\geq 0}, ~h^\top \lambda \leq b(z),~ H^\top \lambda = a(z) \},
    \end{align*}
    be the feasible set of Problem \eqref{eq:robust_LP} and the feasible set of Problem \eqref{eq:dual_robust_LP_formulation} projected onto its first $d$ coordinates, respectively. 
    Then we have that $\mathcal{Z} = \mathcal{Z}'$.
    \label{prop:main_result_robust_LP}
\end{prop}

\subsection{Scenario optimisation}
\label{subsec:prelim_scenario_approach}
The scenario approach theory establishes sample complexity guarantees for the probability of constraint violation of a  chance-constrained optimisation problem \cite{CG:08}.   
Let $(\uncertaintyspace, \sigmaalgebra, \probmeas)$ be a probability space, where $\uncertaintyspace$ is the sample space, $\sigmaalgebra$ is a $\sigma$-algebra over $\uncertaintyspace$, and $\probmeas$ is a probability measure  over $\sigmaalgebra$. Then,  a chance-constrained  program is defined as:
\begin{equation}
\label{Problem:chanceConstrainedProblem}
    \begin{aligned}
        \min_{z} & \quad s^\top z \\
        \subjectto & \quad \probmeas\{\uncertaintyelement \in \uncertaintyspace : g(z, \uncertaintyelement) \leq 0\} \geq 1-\epsilon, 
    \end{aligned}
\end{equation}
where $z\in \mathbb{R}^d$ is the optimisation variable, $s \in \mathbb{R}^d$ are the cost coefficients, $g : \mathbb{R}^d \times \uncertaintyspace \to \mathbb{R}$ is a function that is convex in $z$ for each value of $\uncertaintyelement$ and measurable in $\uncertaintyelement$ for each value of $z$, and $\epsilon \in (0, 1)$ is a given bound on constraint violation.

Now assume $\sampleset = \{\uncertaintyelement_1, \ldots, \uncertaintyelement_N$\} is a set of independent samples from $\probmeas$. Note that the set $\sampleset$ belongs to the space $(\uncertaintyspace^N, \otimes_N \sigmaalgebra, \probmeas^N)$, where $\uncertaintyspace^N$ is the $N$-fold Cartesian product of $\uncertaintyspace$, and $\otimes_N \sigmaalgebra$ is the product $\sigma$-algebra generated by the $\sigma$-algebra $\sigmaalgebra$ and $\probmeas^N$ represents the induced measure on $\uncertaintyspace^N$ \cite{CG:08}. Then,  at the core of the scenario approach is the construction of the scenario program
\begin{equation}
    \begin{aligned}
    \minimise_{z} & \quad s^\top z \\
    \subjectto & \quad g(z,\uncertaintyelement) \leq 0, \quad \text{ for all } \uncertaintyelement \in \sampleset.
    \end{aligned}
    \label{eq:scenario_program}
\end{equation}
The idea of the scenario approach is to use Problem \eqref{eq:scenario_program} to obtain high confidence bounds on the solution of Problem \eqref{Problem:chanceConstrainedProblem}.
To do that, we need some standard assumptions \cite{CG:08}.
\begin{assump}
    Assume that:
    \begin{itemize}
        \item $\probmeas^N$-almost surely, the feasible set of Problem \eqref{eq:scenario_program} given by
        $
            \mathcal{Z} = \{ z \in \mathbb{R}^d : g(z,\uncertaintyelement) \leq 0, \text{ for all } \uncertaintyelement \in \sampleset \},
        $ 
        has non-empty interior.
        \item $\probmeas^N$-almost surely, the optimal solution of Problem \eqref{eq:scenario_program} exists and is unique.
    \end{itemize}
    \label{assump:well-posedness-scenario}
\end{assump}

Denote by $z^\star(\sampleset)$ the unique, optimal solution of Problem \eqref{eq:scenario_program}, which is a random variable from $\uncertaintyspace^N $ to $\mathbb{R}^d$. Then, we are ready to state Proposition \ref{Prop:Scenario}. 
\begin{prop}[\cite{CG:08}]
\label{Prop:Scenario}
Let $N \in \mathbb{N}$ represent the number of available samples and
$
    V(z) = \probmeas\{\uncertaintyelement \in \uncertaintyspace : g(z, \uncertaintyelement) > 0\}.
$ be the probability of constraint violation associated with $z^\star(\sampleset)$. Assume a threshold $\epsilon \in (0,1)$ is given.  Then we have that
\[
\probmeas^N \{ \sampleset \in \uncertaintyspace^N : V(z^\star(\sampleset)) > \epsilon \} \leq \sum_{i = 0}^{d-1} \binom{N}{i} \epsilon^i (1-\epsilon)^{N-i}.
\]
\label{prop:scenario_approach_theory}
\end{prop}
Proposition \ref{prop:scenario_approach_theory} will be key in establishing safety guarantees for the class of stochastic models considered in this paper. 

%% file: sections/problem_statement.tex
The goal of this paper is to certify safety for \acrlong{pwa} stochastic systems, which we formally introduce in Section \ref{sec:PWASystemsDefinition}, while probabilistic safety is introduced in Section \ref{subsec:pro_formulation_safety}.

\subsection{Piece-wise affine stochastic systems}
\label{sec:PWASystemsDefinition}
Let $\regionset = \{\region_1, \ldots, \region_\numberregions\}$ be a polyhedral partition of the state space $X \subseteq \mathbb{R}^n$, where each $\region_i$, $i=1,\ldots,\numberregions$ is given by its half-space representation.
Consider the following discrete-time stochastic \gls{pwa} system: 
\begin{equation}
    \label{eq:system}
    \mathbf{x}(k+1) = f(\mathbf{x}(k)) + \noise(k), \quad \mathbf{x}(0) \in \initialset,
\end{equation}
where $k \in \mathbb{N}$ denotes the (discrete) time index, $\initialset $ is a set of initial states, and $f:X \mapsto \mathbb{R}^n$ is a \gls{pwa} vector field given by
\[ f(x) = f_i(x) = A_i x + b_i, \quad x \in \region_i \subseteq \mathbb{R}^n, \]
for some matrix $A_i \in \mathbb{R}^{n \times n}$ and vector $b_i \in \mathbb{R}^{n}$.
The additive term $(\noise(k))_{k \in \mathbb{N}}$
is a sequence of independent and identically distributed random variables representing an additive noise term. We assume that $\noise(k)$ is defined on the filtered probability space $(\uncertaintyspace, \sigmaalgebra, (\sigmaalgebra_k)_{k \in \mathbb{N}},\mathbb{P})$, where $\sigmaalgebra_{k}$ is the natural filtration of the process $\noise(k)$, and $\mathbb{P}$ is assumed to be unknown. Consequently, $(\mathbf{x}(k))_{k \in \mathbb{N}}$ is also a stochastic process in the space $(\uncertaintyspace, \sigmaalgebra, \probmeas)$ that is, it is $\sigmaalgebra_{k-1}$-measurable \cite{Shapiro2021-zc}.
We note that System \eqref{eq:system} represents a flexible and expressive model. In fact, not only does it include linear systems, but we note that \gls{pwa} functions can approximate any non-linear function arbitrarily well.

\subsection{Time-bounded probabilistic safety}
\label{subsec:pro_formulation_safety} 
Our goal is to study probabilistic safety for System \eqref{eq:system}. 
\begin{defi}[Probabilistic safety \cite{PJP:07}]\label{defi:weak_FI}
Let $T \in \mathbb{N}$ be a time horizon and $\safeset$ be a measurable subset of $X$\footnote{If $X \neq \mathbb{R}^n$ then it may be necessary to replace $\mathbf{x}(k)$ with an equivalent stopped process $\tilde{\mathbf{x}}(k)$ \cite{Santoyo2021}.}. We define probability safety\footnote{
Notice that $\eta(k)$ is measurable function $\mathbb{N} \times \uncertaintyspace \to \mathbb{R}^n$, omitting the dependence on $\uncertaintyelement \in \uncertaintyspace$. Therefore, when we use the notation $\probmeas\{ \uncertaintyelement \in \uncertaintyspace : \mathbf{x}(k) \in \safeset, \text{ for all } k \in \{1, \ldots, T\}\}$, the reader should have in mind that the process $\mathbf{x}$ is dependent on $\uncertaintyelement$. Please refer to \cite{Shapiro2021-zc} for more details.
} for System \eqref{eq:system} as
\begin{equation}
\zeta(\safeset,T) = \mathbb{P}\{ \uncertaintyelement \in \uncertaintyspace : \mathbf{x}(k) \in \safeset \text{ for all } k \in \{0, \ldots, T\} \}.
\end{equation}
\end{defi}

We assume that, while $f$ is known, $\eta$ is unknown and we can only generate \gls{iid} samples from it. Under these assumptions, the goal in this paper is to compute a (non-trivial) lower bound on $\zeta(\safeset,T)$ for System \eqref{eq:system}.

Our approach is based on using the sampled data to synthesize a \acrfull{pwa} \acrfull{sbf} for System \eqref{eq:system} with high confidence. 
In order to do that, in Section \ref{subsec:data_driven_sbf_infinite_representation} we develop a novel and powerful inner approximation for the feasible set of stochastic programs in terms of chance-constrained optimisation. This result is employed to use the scenario approach to synthesize \gls{sbf} for System \eqref{eq:system} with high confidence and by requiring a number of data logarithmic in the negative inverse of the confidence.
In Section \ref{subsec:data_driven_sbf_finite_representation}, we show that in the setting considered in this paper the resulting optimization problem reduces to \gls{lp}, thus enabling efficient and scalable synthesis.
Before presenting the main result, we review in the next section \glspl{sbf} and how they can be employed to guarantee a lower bound on $\zeta(\safeset,T)$.

%% file: sections/stochastic_barrier_functions.tex
\Glspl{sbf} are Lyapunov-like functions commonly employed to compute the safety probability of stochastic systems \cite{PJP:07}. 
\begin{defi}[\acrlong{sbf}]\label{defi:Barrier_certificate}
Let $\unsafeset = X \setminus \safeset$ be the unsafe set and $\initialset \subseteq \safeset$ the set of initial states, then a non-negative function $B: X \mapsto \mathbb{R}_{\geq 0}$ is called a \acrlong{sbf} if there exist non-negative constants $\gamma, \cmartingale$ satisfying the following conditions
\begin{align}
    B(x) \leq \gamma, &~ \text{ for all } x \in \initialset,\label{eq:initial_set_constraint}\\
    B(x) \geq 1, & ~ \text{ for all }x \in \unsafeset,
\end{align}
\begin{equation}
    \mathbb{E} \left[ B(f(x) + \noise(\uncertaintyelement)) \right] \leq B(x) + \cmartingale, ~ \text{ for all } x \in \safeset. \label{eq:c-martingale}
\end{equation}
where the expectation is with respect to $\uncertaintyelement \sim \probmeas$.
\end{defi}

\begin{figure}
    \centering
    \includegraphics[width=\linewidth]{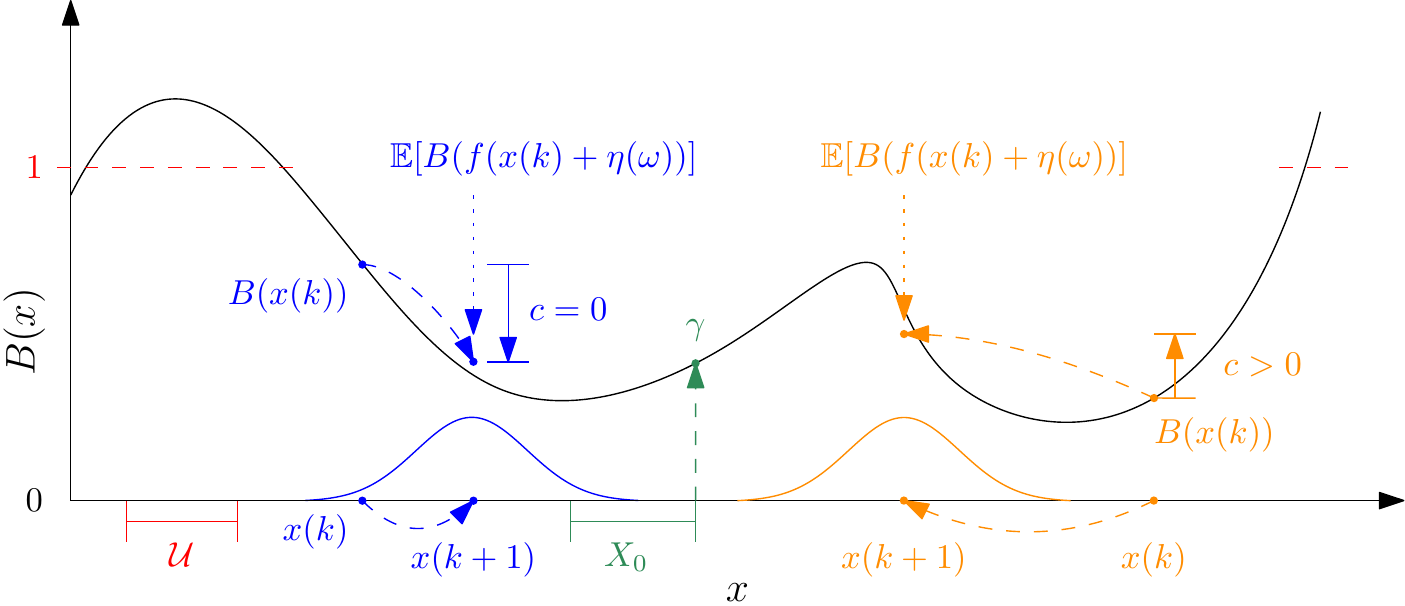}
    \caption{The figure is borrowed from \cite{Mathiesen2013}. A \gls{sbf} $B(x)$ is a non-negative function that is greater than $1$ in an unsafe region $\unsafeset$, which is the complement of the safe set $\safeset$. The variable $\gamma$ is an upper bound for $B(x)$ over an initial region $\initialset$. The upper bound for the expected increase in $B(x)$ after one step of \eqref{eq:system} over the safe set $\safeset$ is denoted $\cmartingale$. Then it holds that the probability of safety $\zeta(\safeset, T) \geq 1 - (\gamma + \cmartingale T)$.}
    \label{fig:example_continuous_barrier_function}
\end{figure}

A pictorial representation of a \gls{sbf} is presented in Figure \ref{fig:example_continuous_barrier_function}.
Intuitively, the conditions in Definition \ref{defi:Barrier_certificate} allows one to use martingale inequalities to lower bound probabilistic safety.

\begin{prop}[{\cite[Chapter 3, Theorem 3]{kushner1967stochastic}}]
Let $B$ be a \gls{sbf} satisfying the conditions in Definition \ref{defi:Barrier_certificate} for System \eqref{eq:system}, time horizon $T$, and safe set $\safeset$. Then, it holds that $\zeta(\safeset,T) \geq 1 - (\gamma + \cmartingale T)$.
\label{prop:barrier_prob_safety}
\end{prop}
Thanks to Proposition \ref{prop:barrier_prob_safety}, a sufficient condition to establish a lower bound on the safety probability is to design a \gls{sbf} satisfying Equations \eqref{eq:initial_set_constraint}-\eqref{eq:c-martingale}. This can be obtained by solving the following stochastic program where $B$ is parameterised by $\theta$ as $B(x, \theta)$ according to a chosen function class
\begin{equation}
    \min\limits_{\gamma,\, \cmartingale,\, \theta} \quad \gamma + \cmartingale T,
    \tag{BP}
	\label{eq:main_opt_pro}
\end{equation}
subject to the conditions in Definition \ref{defi:Barrier_certificate}\footnote{For all barrier programs, we use an abbreviated reference to carry semantic meaning about the variation, such as Problem \eqref{eq:main_opt_pro} for the general barrier program.}. In other words, synthesis of a \gls{sbf} can be framed as a minimisation over $\gamma + \cmartingale T$.
In this optimisation problem, the expectation condition (Equation \eqref{eq:c-martingale}) can generally be computed analytically only under some  strong assumptions on the noise distribution \cite{Jagtap2020, Santoyo2021}.
Our approach proposes a new, inner chance-constrained approximation of Problem \eqref{eq:main_opt_pro}, which allows us to rely on tools from scenario optimisation to synthesize a barrier \cite{CG:08}. The resulting approach is a distribution-free, data-driven method to obtain a \gls{sbf} as a safety certificate with a high confidence of validity.
Note that to guarantee the convexity of Problem \eqref{eq:main_opt_pro}, $B$ is generally restricted to be either a SoS polynomial or an exponential function \cite{Santoyo2021}. In this paper, motivated by the structure of System \eqref{eq:system}, we will consider \acrlong{pwa} $B$, which have the flexibility to be able to model arbitrarily well any continuous function assuming the number of pieces of $B$ is large enough.

%% file: sections/data_driven_sbf_design.tex
In this section, we present the main results of this paper. In Section \ref{subsec:data_driven_sbf_infinite_representation}, an inner approximation of the feasible set of Problem \eqref{eq:main_opt_pro} in terms of a chance-constrained problem (Theorem \ref{thm:barrier-ccp}) is described. Such a relaxation allows us to use the scenario approach to derive high confidence bounds on the resulting solution (Corollary \ref{corollary:scenario-barrier-design}). Finally
in Section \ref{subsec:data_driven_sbf_finite_representation}, we will introduce \gls{pwa} \glspl{sbf} and show how for this class of barriers the resulting scenario approach is a \gls{lp}.
\subsection{Data-driven stochastic barrier design}\label{subsec:data_driven_sbf_infinite_representation}
Solving Problem \eqref{eq:main_opt_pro} is challenging because analytic expressions of the expectation constraint are rarely available, even if the distribution $\probmeas $ is known (which is not the case in this paper).
To solve this problem, in Theorem \ref{thm:barrier-ccp}, we derive a chance-constrained problem whose feasible set is a subset of the feasible set for Problem \eqref{eq:main_opt_pro}. Thus, its optimal solution is an upper bound to that of Problem \eqref{eq:main_opt_pro}. 

\begin{theorem}
    Consider System \eqref{eq:system}, and the barrier function $B(x, \theta)$ as in Definition \ref{defi:Barrier_certificate}, where $B$ is convex in $\theta$. Assume a given $\epsilon \in (0,1)$ and $M \geq 1$, and define decision variables $z = (\cmartingale,\gamma,\theta)$. Let $g(x, z, \noise(\uncertaintyelement)) = B(f(x) + \noise(\uncertaintyelement),\theta)$ and $h(x,z) = B(x,\theta) + \cmartingale$, and choose $\buffervar \geq \frac{\epsilon M}{1-\epsilon}$. Define the set
    \[
        E(x, z) = \left\{ \uncertaintyelement \in \uncertaintyspace : g(x, z, \noise(\uncertaintyelement)) + \buffervar \leq h(x, z)\right\}.
    \]
    Then, the feasible set of the chance-constrained barrier program
	\begin{equation}
		\begin{aligned}
			\min_{\gamma \geq 0,\, \cmartingale \geq 0,\, \theta} & \quad \gamma + \cmartingale T\\
			\subjectto & \quad \begin{aligned}
				&\begin{alignedat}{3}
					& B(x,\theta) \in [0, M], &&\quad \text{ for all } x\in \mathbb{R}^n, \\[1ex]
					&B(x,\theta) \leq \gamma, &&\quad \text{ for all } x \in \initialset, \\[1ex]
					&B(x,\theta) \geq 1, &&\quad \text{ for all } x \in \unsafeset, \\[1ex]
				\end{alignedat}\\
				&\mathbb{P}\left\{E(x, z) \right\} \geq 1 - \epsilon, \text{ for all } x \in \safeset,
			\end{aligned}
		\end{aligned}
		\tag{CCBP}
		\label{eq:ccp_opt_pro}
	\end{equation}
	is contained in the feasible set of Problem \eqref{eq:main_opt_pro}.	
	\label{thm:barrier-ccp}
\end{theorem}
The proof of Theorem \ref{thm:barrier-ccp} is reported in Section \ref{sec:technical_proofs}.
Theorem \ref{thm:barrier-ccp} opens new ways for data-driven design of \acrlongpl{sbf}.
Rather than relying on standard concentration inequalities to approximate the expectation in Equation \eqref{eq:c-martingale} as in \cite{SALAMATI20217}, we can perform chance-constraint tightening with the parameter $\buffervar$ to guarantee the feasible set of \eqref{eq:barrier_sp_opt_pro} is an inner approximation of \eqref{eq:ccp_opt_pro}. Building on this result, in Corollary \ref{corollary:scenario-barrier-design} we use the scenario approach to design \glspl{sbf} from data with high confidence. 

\begin{corollary}
	Assume that $\sampleset = \{ \uncertaintyelement_1, \ldots, \uncertaintyelement_N \}$ is a collection of $N$ independent samples from the distribution $\probmeas$. Fix $\epsilon \in (0,1)$, $M \geq 1$ and $\buffervar \geq \frac{\epsilon M}{1-\epsilon}$, and let $\beta = \sum_{i = 0}^{d-1} \binom{N}{i} \epsilon^i (1-\epsilon)^{N-i},$
	where $d = |\theta| + 2$. Let $(\cmartingale^\star,\gamma^\star,\theta^\star)$ be the optimal solution to the scenario program 
	\begin{equation}
		\begin{aligned}
			\min_{\gamma \geq 0, \, \cmartingale \geq 0,\, \theta} & \quad \gamma + \cmartingale T\\
			\subjectto
				&\quad B(x, \theta) \in [0, M], && \text{ for all } x \in \mathbb{R}^n, \\[1ex]
				&\quad B(x, \theta) \leq \gamma, && \text{ for all } x \in \initialset , \\[1ex]
				&\quad B(x, \theta) \geq 1, && \text{ for all } x \in \unsafeset, \\[1ex]
				&\quad g(x, z, \noise(\uncertaintyelement)) + \buffervar \leq h(x, z), && \\
				   & \qquad \text{ for all } \uncertaintyelement \in \sampleset, &&\text{ for all } x \in \safeset,
		\end{aligned}
		\tag{SBP}
		\label{eq:barrier_sp_opt_pro}
	\end{equation}
	where $g$ and $h$ are defined as in Theorem \ref{thm:barrier-ccp}. Then, with confidence $1-\beta$, it holds that
  $\zeta(\safeset,T) \geq 1 - (\gamma^\star + \cmartingale^\star T)$.
	
	\label{corollary:scenario-barrier-design}
\end{corollary}

\begin{remark}
\label{remark:DataEfficiency}
Observe that the amount of data $N$ required to achieve a desired confidence $1-\beta$ with existing approaches based on concentration inequalities to approximate Equation \eqref{eq:c-martingale} is proportional to $1/\beta$ \cite{SALAMATI20217} whereas for our approach, the amount required is proportional to $\ln(1/\beta)$ \cite{Campi2009a}. To put this into perspective, consider $\beta = 10^{-9}$, which is the gold standard in both aviation and autonomous vehicle design \cite{Shalev-Shwartz2017}, then $1/\beta = 10^9$ while $\ln(1/\beta) \approx 20.7$.
\end{remark}

\subsection{Linear programming reformulation of stochastic barrier function design}
\label{subsec:data_driven_sbf_finite_representation}
Corollary \ref{corollary:scenario-barrier-design} defines an optimisation problem (Problem \eqref{eq:barrier_sp_opt_pro}) for the data-driven design of \glspl{sbf}. For instance, the resulting problem can be solved under the assumption that $B$ is a \gls{sos} function using semi-definite programming \cite{PJP:07, Santoyo2021}. However, while viable, this approach can often be conservative and lack of scalability \cite{Mathiesen2013}. Motivated by the \gls{pwa} structure of System \eqref{eq:system}, we propose instead to use a \gls{pwa} function to parameterise a \gls{sbf}. Then, by applying tools from robust \gls{lp}, i.e. Proposition \ref{prop:main_result_robust_LP}, we show that Problem \eqref{eq:barrier_sp_opt_pro} can be transformed into a linear program with a finite number of constraints.
To this end, let $\bar{\regionset} = \{ \bar{\region}_1, \ldots, \bar{\region}_{\bar{\numberregions}} \}$ be a polyhedral partition of the state space $X$ with $\bar{\numberregions} \geq \numberregions$. We assume that for any two regions $i, j$ where $i \neq j$ the the intersection has zero-measure \[\probmeas \{\uncertaintyelement \in \uncertaintyspace : \mathbf{x}(k) \in \bar{\region}_i \cap \bar{\region}_j \text{ for all }k = 0,\ldots, T\} = 0.\] Furthermore, assume for simplicity that each region $\bar{\region}_i$ is a subset of exactly one region $\region_{r(i)}$ from the partition $\regionset$, with a surjective function $r : \{1, \ldots, \bar{\numberregions}\} \to \{1, \ldots, \numberregions\}$ mapping between indices. In other words, the partition for the \gls{pwa} barrier candidate $\bar{\regionset}$ is aligned with the partition of the dynamics $\regionset$, although potentially more fine-grained. We consider a \gls{pwa} \gls{sbf} $B$ defined as follows
\begin{equation}\label{eq:barrier_family_def}
    B(x,\theta) = \max(B_1(x,\theta), \ldots, B_{\bar{\numberregions}}(x,\theta)),
\end{equation}
where
\[
B_i(x,\theta) = \left\{ \begin{matrix}
	\barrierl_i^\top x + \barrierc_i,& \quad \text{ for } x \in \bar{\region}_i, \\
	0, & \quad \text{ otherwise},
\end{matrix} \right.
\]
and $\theta \in \mathbb{R}^{\bar{\numberregions} (n+1)}$ is the set of parameters $(u_i,v_i) \in \mathbb{R}^{n+1}$, $i = 1, \ldots, \bar{\numberregions}$, used to define the \gls{sbf}.

For convenience, we also define collections of indices from $I = \{1,\ldots,\bar{\numberregions}\}$ that correspond to elements of the partition $\bar{\regionset}$ that have non-empty intersection with the set of safe, unsafe, and initial states, respectively:
\begin{equation}
\begin{aligned}
	I_\safeset &= \{i \in I : \bar{\region}_i \cap \safeset \neq \emptyset \}, \\ 
	I_\unsafeset &= \{i \in I : \bar{\region}_i \cap \unsafeset \neq \emptyset \}, \\ 
	I_\initialset &= \{i \in I : \bar{\region}_i \cap \initialset \neq \emptyset \}.
\end{aligned}
\label{eq:indices-barrier}
\end{equation}

With the family of barrier functions defined, we turn our attention to the reduction of Problem \eqref{eq:barrier_sp_opt_pro} into a linear problem.
In order to do that we need to reduce each of the constraints in Problem \eqref{eq:barrier_sp_opt_pro} into linear constraints.
The reduction for the non-negativity, upper bound, initial, and unsafe set constraints follow a similar structure. Hence, for brevity, we only describe the process for the non-negativity constraint.
With the assumption that the intersection of two regions has no volume, we can impose $B_i(x, \theta) \geq 0$ for all $x \in \bar{\region}_i$ independently for each region.
Note that for each region $i$ the barrier $B_i(x, \theta)$ is an affine function in $x$ over the polyhedron $\bar{\region}_i$. Hence, the resulting constraint is a robust \gls{lp} constraint and we can rely on Proposition \ref{prop:main_result_robust_LP} to transform the problem to a lifted space representable by a regular \gls{lp} constraint. 
More concretely, consider the constraint $B_i(x, \theta) = \barrierl_i^\top x + \barrierc_i \geq 0$ for all $x \in \bar{\region}_i$ where $\bar{\region}_i$ is defined by its half-space representation $(H_i, h_i) \in \mathbb{R}^{m \times n} \times \mathbb{R}^m$. Then with a dual variable $\lambda_i \in \mathbb{R}^m_{\geq 0}$, this can be replaced with the following two equivalent constraints using Proposition \ref{prop:main_result_robust_LP}: $h_i^\top \lambda_i \leq \barrierc_i$ and $H_i^\top \lambda_i = - \barrierl_i$.

Now, consider the last constraint of Problem \eqref{eq:barrier_sp_opt_pro}, namely $g(x, z, \noise(\uncertaintyelement)) + \buffervar \leq h(x, z)$ for all $\uncertaintyelement \in \sampleset$, for all $x \in \safeset$. 
For this constraint Proposition \ref{prop:main_result_robust_LP} is not immediately applicable, as we must consider the value of the barrier before and after a transition.
Instead, we construct a robust \gls{lp} constraint for each pair of regions $(i,j) \in I_\safeset \times I$:
\begin{equation}
    \begin{aligned}
        &B_j(f_{r(i)}(x) + \noise(\uncertaintyelement)) + \buffervar \leq B_i(x) + \cmartingale,\\ &\qquad\qquad \text{ for all }\uncertaintyelement \in \sampleset, \text{ for all }x \in Q_{ij}(\uncertaintyelement).
    \end{aligned}
	\label{eq:split-semi-infinite-one-step-constraints}
\end{equation}
The random subset $Q_{ij}(\uncertaintyelement)$ of $X$ is defined as
\begin{equation}
	Q_{ij}(\uncertaintyelement) = \{ x \in \bar{\region}_i :  f_{r(i)}(x) + \noise(\uncertaintyelement) \in \bar{\region}_j \},
	\label{eq:random-subset-partition}
\end{equation}
representing the set of elements in the region $\bar{\region}_i$ that are mapped to $\bar{\region}_j$ under a given realisation of the noise $\uncertaintyelement$. A pictorial example of $Q_{ij}(\uncertaintyelement)$ can be found in Figure \ref{fig:preimage_intersection}.
Since both $\bar{\region}_i$ and $\bar{\region}_j$ are polyhedra and $f_{r(i)}$ is an affine function, $Q_{ij}(\uncertaintyelement)$ is a polyhedron \cite{BV:04}. Thus, we can again use Proposition \ref{prop:main_result_robust_LP} to transform Equation \eqref{eq:split-semi-infinite-one-step-constraints} to linear constraints.
Specifically, for a pair of regions $(i, j) \in I_\safeset \times I$ and a realisation of the noise $\uncertaintyelement \in \sampleset$, with half-space representation $(H_{ij\uncertaintyelement}, h_{ij\uncertaintyelement}) \in \mathbb{R}^{m \times n} \times \mathbb{R}^m$ of region $Q_{ij}(\uncertaintyelement)$ and dual variable $\dualvariable_{ij\uncertaintyelement} \in \mathbb{R}_{\geq 0}^m$, the original semi-infinite constraint is transformed into the following two constraints 
\[
    \begin{aligned}
        & h_{ij\uncertaintyelement}^\top \dualvariable_{ij\uncertaintyelement} \leq \barrierc_i - \barrierc_j - \barrierl_j^\top (b_{r(i)} + \noise(\uncertaintyelement)) + \cmartingale - \buffervar, \\[1ex]
        & H^\top_{ij\uncertaintyelement} \dualvariable_{ij\uncertaintyelement} = A_{r(i)}^\top \barrierl_j - \barrierl_i.
    \end{aligned}
\]

Collecting together all finite sets of constraints, the
LP equivalent representation of Program \eqref{eq:barrier_sp_opt_pro} is as follows.
\begin{equation}
\begin{aligned}
    \min_{\gamma \geq 0, \, \cmartingale \geq 0,\, \theta} & \quad \gamma + \cmartingale T\\
    \subjectto & \quad \begin{alignedat}{3}
        & h_i^\top \dualvariable_{i} \leq \barrierc_i, \;  H_i^\top \dualvariable_i = -\barrierl_i,  \\[1ex]
        & h_{i}^\top \dualvariable_{iM} \leq M - \barrierc_i, \; H_{i}^\top \dualvariable_{iM} = \barrierl_i, \text{ for all } i \in I, \\[1ex]
        & h_{i0}^\top \dualvariable_{i0}\leq \gamma - \barrierc_i, \;H^\top_{i0} \dualvariable_{i0} = \barrierl_{i0}, \text{ for all } i \in I_\initialset,  \\[1ex]
        & h_{i}^\top \dualvariable_{i\unsafeset}\leq \barrierc_i - 1, \; H^\top_{i} \dualvariable_{i\unsafeset} = -\barrierl_{i}, \text{ for all } i \in I_\unsafeset, \\[1ex]
        & h_{ij\uncertaintyelement}^\top \dualvariable_{ij\uncertaintyelement} \leq \barrierc_i - \barrierc_j - \barrierl_j^\top (b_{r(i)} + \noise(\uncertaintyelement)) + c - \buffervar, \\[1ex]
        & H^\top_{ij\uncertaintyelement} \dualvariable_{ij\uncertaintyelement} = A_{r(i)}^\top \barrierl_j - \barrierl_i, \text{ for all } \uncertaintyelement \in \sampleset, \\
        & \quad \text{ for all } (i, j) \in I_\safeset \times I,
    \end{alignedat}
\end{aligned}
\tag{LBP}
\label{eq:barrier_sp_opt_pro_finite}
\end{equation}
where $\dualvariable_{i}, \dualvariable_{iM}, \dualvariable_{i0}, \dualvariable_{i\unsafeset}, \dualvariable_{ij\uncertaintyelement}$ are non-negative dual variables.
$(H_{i0}, h_{i0})$ denotes the half-space representation of $\bar{\region}_i \cap \initialset$.
\begin{theorem}\label{thm:finite-scenario-barrier-design}
    Let $B$ be a \acrlong{pwa} \acrlong{sbf} as defined in Equation \eqref{eq:barrier_family_def}. Then, an optimal solution $z^\star(\sampleset)$ to Problem \eqref{eq:barrier_sp_opt_pro_finite} is an optimal solution to Problem \eqref{eq:barrier_sp_opt_pro}.
\end{theorem}
By Corollary \ref{corollary:scenario-barrier-design} and Theorem \ref{thm:finite-scenario-barrier-design}, Problem \eqref{eq:barrier_sp_opt_pro_finite} is an equivalent \gls{lp} representation of Problem \eqref{eq:main_opt_pro} that can be employed to synthesize a SBF.
The number of decision variables and constraints of the resulting LP depends on the number of half-spaces necessary to represent each polyhedron. In particular, assume for simplicity that 
each polyhedral region is represented by $m$ half-spaces. Then, the number of decision variables in Problem \eqref{eq:barrier_sp_opt_pro_finite} is
\[
    \underbrace{2}_{\gamma, \cmartingale} + \underbrace{(n + 1) \cdot \bar{\numberregions}}_{\theta} + \underbrace{m\cdot (2 \bar{\numberregions} + \lvert I_\initialset \rvert + \lvert I_\unsafeset \rvert + N \lvert I_\safeset \rvert \bar{\numberregions})}_{\text{dual variables}},
\]
while the number of constraints is:
\[
    2 + m \cdot (6 \bar{\numberregions} + 3 \lvert I_\initialset \rvert + 3 \lvert I_\unsafeset \rvert + 3 N \lvert I_\safeset \rvert \bar{\numberregions}).
\]
Note that both the number of constraints and number of variables are dominated by the term $m N \lvert I_\safeset \rvert \bar{\numberregions}$, where $ \lvert I_\safeset \rvert $ and $\bar{\numberregions}$ are respectively number of pieces in the \gls{sbf} that intersect with $\safeset$ and total number of pieces in the \gls{sbf}. This illustrates how the dimension of the resulting \gls{lp} problem grows linearly in the number of samples $N$ and quadratically in the complexity (i.e., number of pieces) of the barrier $B$.

\begin{figure}
    \centering
    \includegraphics[width=0.44\textwidth]{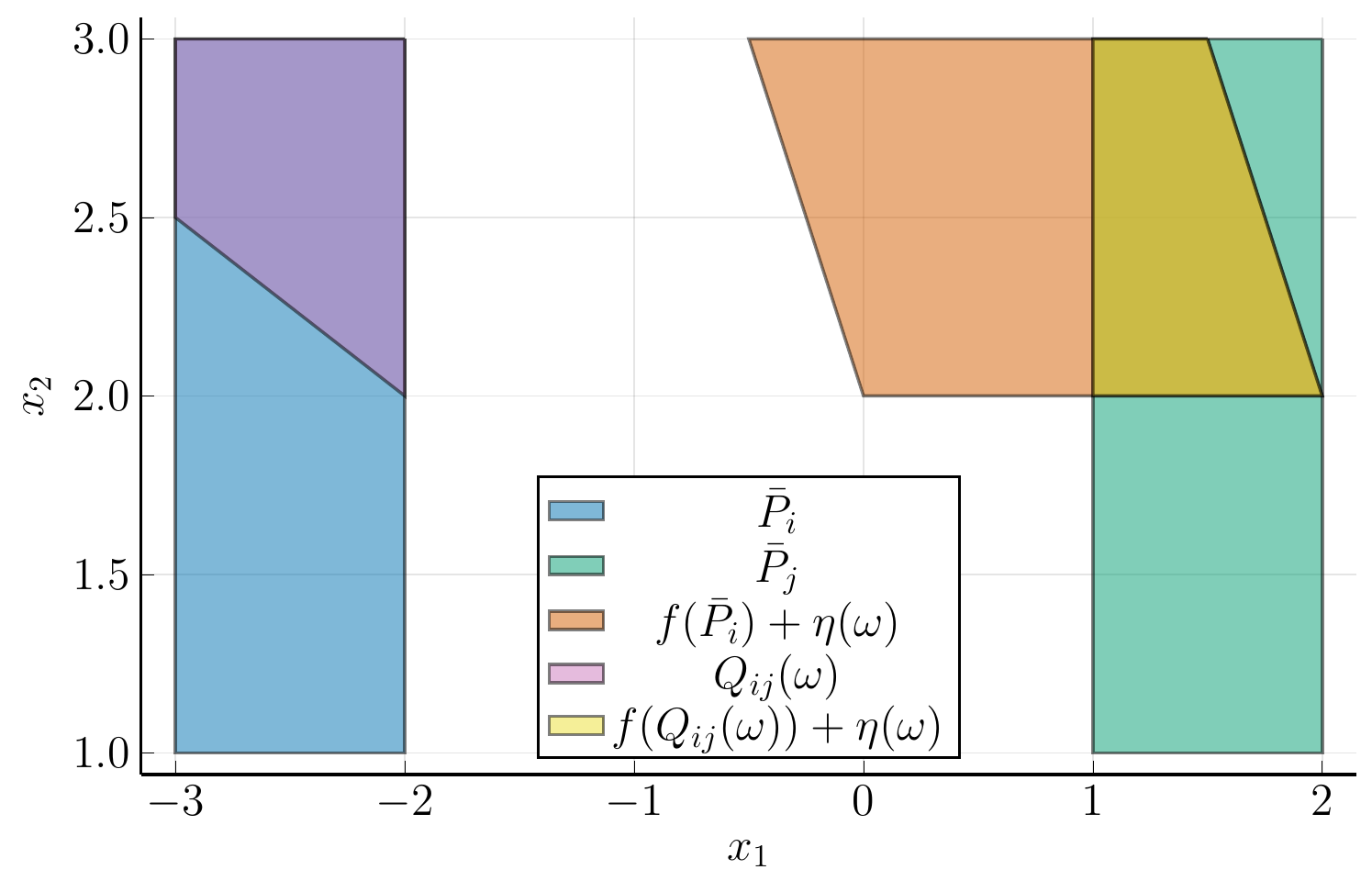}
    \caption{Given two regions $\bar{\region}_i, \bar{\region}_j$ and a realisation of the noise $\uncertaintyelement$, the set $Q_{ij}(\uncertaintyelement)$ represents the subset of $x\in \bar{\region}_i$ such that $f(x) + \noise(\uncertaintyelement) \in \bar{\region}_j$. In other words, $Q_{ij}(\uncertaintyelement)$ is the subset of $\bar{\region}_i$ that can reach $\bar{\region}_j$ given the realisation of the noise $\uncertaintyelement$.}
    \label{fig:preimage_intersection}
\end{figure}

%% file: sections/experiments.tex
To show the efficacy of the proposed method, we evaluate it on three different benchmarks. Namely:
\begin{itemize}
    \item a 1D linear system governed by the following dynamics $\mathbf{x}(k + 1) = \mathbf{x}(k) + \noise(k)$, which is a martingale,
    \item a 2D linear model of longitudinal dynamics for a drone from \cite{Badings2022},
    \item a 2D \gls{pwa} model of a vehicle driving with constant velocity subject to a wind disturbance along its path. 
\end{itemize}
For the martingale system, the goal is to quantify the probability that from any state within a radius of 0.5 around the origin the system will stay within a set of radius of 2.5 around the origin for a time horizon $T = 10$. For the drone, the goal is to certify that the speed of the drone always stays lower than $10$ units, again for a time horizon $T=10$. Please note that in \cite{Badings2022}, they consider an uncertain mass of the drone, which is not compatible with Problem \eqref{eq:barrier_sp_opt_pro_finite}. To make the benchmark compatible, we let the mass be equal to the center of the uncertainty interval, namely $m = 1$.
Finally, the last model represents a vehicle driving with constant velocity. The goal is to stay on the road within $T=10$, despite a varying disturbance from wind along the route.
Mathematically, we describe the dynamics as follows:
\[
    \mathbf{x}(k + 1) = \begin{bmatrix}
        1 & 0 \\ 0 & 0.95\tau
    \end{bmatrix}\mathbf{x}(k) + \begin{bmatrix}
        v\tau \\ 0.5 d \tau^2
    \end{bmatrix} + \eta(k)
\]
where we choose a velocity $v = 13.89$, a time resolution $\tau = 1$, and a disturbance $d = 0.0626$ for regions where the longitudinal position $x_1$ satisfies $80 \leq x_1 \leq 120$ and $d = 0$ otherwise. For the purpose of the experiment, we assume $\eta(k)$ is Gaussian noise with diagonal covariance, which of course is assumed unknown and only \gls{iid} samples can be generated from it.

We compare our method against \gls{saa} \cite{SALAMATI20217}, arguably the state-of-the-art for data-driven synthesis of \glspl{sbf}, on the three benchmarks. For \gls{saa}, we employ a 4th degree polynomial barrier and \gls{sos} optimisation.
For our method, we consider a \gls{pwa} barrier function with $7$ and $33$ pieces for respectively the Martingale and Drone example, while for the Vehicle example we consider different values of $\bar{\numberregions}$ to study its impact.
The benchmarks and methods have been implemented\footnote{Code is available at \url{https://github.com/DAI-Lab-HERALD/scenario-barrier} under a GNU GPLv3 license.} in Julia (1.8.3) with JuMP.jl (1.6.0) as the modelling framework and Mosek (9.3.11) as the \gls{lp} solver. 
The experiments are conducted on a computer running Linux Manjaro (5.10.157) with an Intel Core i7-10610U CPU and 16GB RAM.

Table \ref{tab:results_table} shows the results across all three systems. The results are reported as the average over 100 trials to ensure that certification is not spurious due to a sampling of the noise.
Comparing the two methods in Table \ref{tab:results_table}, we see that the proposed method outperforms \gls{saa} across all measures on both the Martingale and Drone system, while the vehicle is intractable for \gls{saa}.  
Note that for any system considered in this paper, \gls{saa} can only certify with a confidence $1 - 10^{-6}$ and probability of safety up to $0.95$, due to an intractable amount of samples required for higher confidence and smaller auxiliary variable $\nu$. On the other hand, our method, thanks to the bounds we compute in Corollary \ref{corollary:scenario-barrier-design}, achieves a confidence of $1 - 10^{-9}$ (see Remark \ref{remark:DataEfficiency}), and a probability of safety up to $0.995$. In addition, our method achieves a higher certified probability of safety and is orders of magnitude faster. The latter is due the reduction to LP and to the use of the scenario approach to derive confidence bounds.
To further highlight the data-efficiency, we present in Figure \ref{fig:number_of_samples_required} the number of samples required to achieve a desired confidence for both methods.
The figure clearly shows that our method requires orders of magnitude less samples to achieve the same confidence.

Next, we analyze the impact of increasing the number of pieces in the \gls{sbf} $\bar{l}$, towards a more expressive \gls{sbf}. Table \ref{tab:results_table} reveals that increasing the number of partitions for the barrier (see Equation \ref{eq:barrier_family_def}) yields tighter guarantees as expected. In fact, a \gls{pwa} function with arbitrarily many pieces can approximate arbitrarily well any continuous function, thus increasing the flexibility of the framework.
 However, this comes at the cost of increased computation time. Note however, that computation times are always faster than \gls{saa} even for relatively large $\bar{\numberregions}$.  We also observe that despite using fewer regions for the Drone system, it is slower to compute than for the Vehicle system with both 42 and 46 regions. To understand why note that the constraint in Equation \eqref{eq:split-semi-infinite-one-step-constraints} is trivially satisfied if $Q_{ij}(\uncertaintyelement)$ is empty, or in other words, it is impossible to reach region $j$ from region $i$ under the realisation of the noise $\uncertaintyelement$. The Drone system has more non-empty $Q_{ij}(\uncertaintyelement)$ over the Vehicle system and thus is slower.

% ("cdc_car_with_sidewind_no_partitioning", 1.0, (0.6183324019375488, 0.002396937346556549), (1.4423357939720154, 0.10815060751299366))
% ("cdc_car_with_sidewind_lat_partitioning", 1.0, (0.842114254897307, 0.001954047665656211), (2.4525604915618895, 0.03183488872935181))
% ("cdc_car_with_sidewind_lon_partitioning", 1.0, (0.7116061586436413, 0.002257720904952243), (3.7249662017822267, 0.12887792778367224))
% ("cdc_car_with_sidewind_latlon_partitioning", 1.0, (0.9940631260208829, 0.0004508643091506448), (9.062844665050507, 0.3720088716709391))
% ("cdc_car_with_sidewind_saa", 1.0, (-0.02643530776064332, 0.000333490878525515), (2.1413737964630126, 0.0370173173892101))
% ("cdc_drone_exact", 1.0, (0.994997498646974, 5.947853320975582e-11), (4.840951838493347, 0.22229616506280792))
% ("cdc_drone_exact_saa", 1.0, (0.9499940958832563, 1.3119049180332354e-9), (1.175646243095398, 0.024135688523603704))
% ("cdc_linear1d", 1.0, (0.7689340058879636, 0.002902222615130512), (0.09579002857208252, 0.007326980500436703))
% ("cdc_linear1d_ssa", 1.0, (0.9097762973355887, 1.0468416036468753e-5), (0.2491968822479248, 0.004760589049658204))

 \aboverulesep=0.3ex
 \belowrulesep=0.3ex
 \setlength{\tabcolsep}{4pt}
 \renewcommand{\arraystretch}{1.1}
\begin{table}
    \centering
    \caption{Certified safety and computation time using the method explained in Section \ref{sec:data_driven_sbf_design}. Results are reported as the average over 100 trials. $n$ is the dimensionality of the system and $\bar{\numberregions}$ is the number of pieces of the \gls{pwa} \gls{sbf} $B$. $ 1 - \beta$ is the confidence in the certificate and $\zeta(\safeset, T)$ is the certified level of safety. Bold font denotes best method for each measure and system.}
    \label{tab:results_table}
    \vspace{0.2em}
    {\footnotesize
    \begin{tabular}{l|clc|ccc}
    \toprule
       System  & $n$ & Method & $\bar{\numberregions}$ & $\beta$ & $\zeta(\safeset, T)$ & Comp. time (s)\\\midrule\midrule
         \multirow{2}{*}{Martingale} & \multirow{2}{*}{1} & Our & 7 & $\mathbf{10^{-9}}$ & $0.769$ & $\mathbf{0.096}$ \\ 
         && \gls{saa} & - & $10^{-6}$ & $\mathbf{0.910}$ & $0.249$ \\\midrule
         \multirow{2}{*}{Drone} & \multirow{2}{*}{2} & Our & 33 & $\mathbf{10^{-9}}$ & $\mathbf{0.995}$ & $4.84$ \\
         && \gls{saa} & - & $10^{-6}$ & $0.950$ & $\mathbf{1.18}$ \\\midrule
         \multirow{5}{*}{Vehicle} & \multirow{5}{*}{2} & Our & 18 & $\mathbf{10^{-9}}$ & $0.618$ & $\mathbf{1.44}$ \\
         && Our & 42 & $\mathbf{10^{-9}}$ & $0.712$ & $2.45$ \\
         && Our & 46 & $\mathbf{10^{-9}}$ & $0.842$ & $3.72$ \\
         && Our & 126 & $\mathbf{10^{-9}}$ & $\mathbf{0.994}$ & $9.06$ \\
         && \gls{saa} & - & $10^{-6}$ & $0.000$ & $2.14$ \\
         \bottomrule
    \end{tabular}}
\end{table}

\begin{figure}
    \centering
    \includegraphics[width=0.48\textwidth]{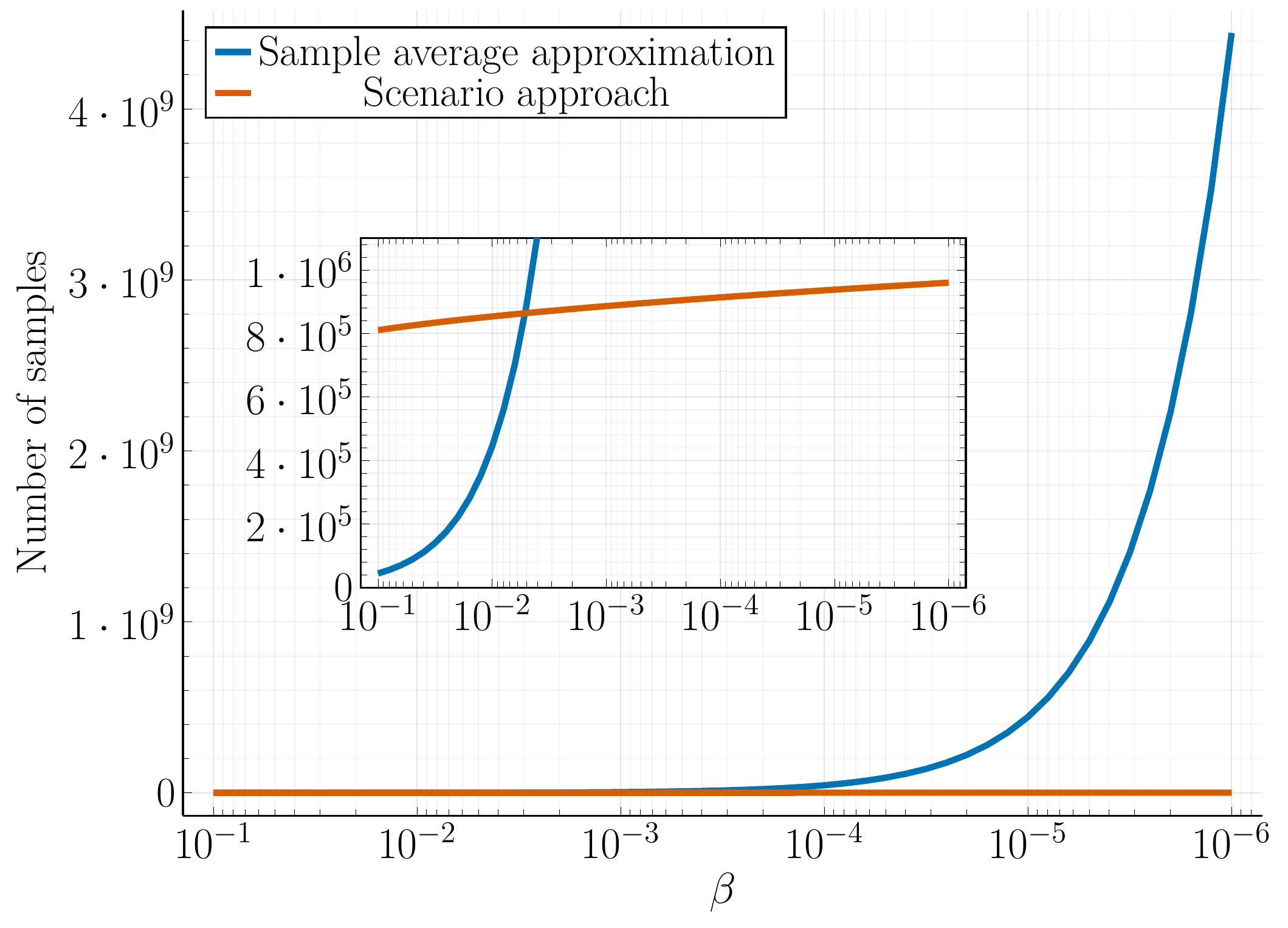}
    \caption{A plot for the number of samples required to achieve a given confidence $1-\beta$ for \gls{saa} and the proposed method using the scenario approach. The number of samples reported in this plot is specifically for the vehicle system with 126 regions, as reported in Table \ref{tab:results_table}.}
    \label{fig:number_of_samples_required}
\end{figure}

%% file: sections/conclusion.tex
We studied the problem of certifying probabilistic safety for partially known stochastic systems.
The problem is important for the adoption of autonomous safety-critical systems.
This safety verification problem was addressed by synthesising \acrfull{sbf} with a data-driven approach leveraging the scenario optimisation theory.
To apply the data-driven scenario approach to \gls{sbf} synthesis, a novel inner chance-constrained approximation to stochastic programming was presented.
The chance-constrained approximation was applied to \glspl{sbf} in Theorem \ref{thm:barrier-ccp}: an important consequence of the theorem is that the method can be easily extended to other classes of systems, e.g. polynomial or more general non-linear systems.
Experimental studies showed that our method can certify systems with a confidence that is orders of magnitude greater than the state-of-the-art methods, while also producing tighter bounds and being faster.

%% file: sections/technical_proofs.tex
\subsection{Proof for Theorem \ref{thm:barrier-ccp}}
In order to prove Theorem \ref{thm:barrier-ccp} we consider the following stochastic program, which generalizes Problem \eqref{eq:main_opt_pro},
\begin{equation}
        \begin{aligned}
        \min_{z} & \quad s^\top z\\
        \subjectto & \quad \mathbb{E} \left\{ g(x, z, \noise(\uncertaintyelement)) \right\} \leq h(x, z), \quad \text{ for all } x \in \safeset, 
        \end{aligned}
    \label{eq:stochastic_opt_pro}
\end{equation}
where $z \in \mathbb{R}^d$ is the decision variable, $s \in \mathbb{R}^d$ is the cost vector, $\noise : \uncertaintyspace \to \mathbb{R}^m$ is a random variable on $(\uncertaintyspace, \sigmaalgebra, \probmeas)$, and $g: \mathbb{R}^n \times \mathbb{R}^d \times \mathbb{R}^m \to \mathbb{R}$ is a measurable and integrable function for each pair $(x, z) \in \mathbb{R}^n \times \mathbb{R}^d$, $h:\mathbb{R}^n \times \mathbb{R}^d \to \mathbb{R}_{\geq 0}$ is a function, and $\safeset$ is a measureable set on $\mathbb{R}^n$.
The feasible set of Problem \eqref{eq:stochastic_opt_pro} is given by
\[
\mathcal{Z} = \{ z \in \mathbb{R}^d : \mathbb{E}\{g(x, z, \noise(\uncertaintyelement))\} \leq h(x, z) \text{ for all } x\in \safeset \}.
\]

Theorem \ref{thm:chance_constrained_inner_approximation} shows that an inner approximation of the feasible set $\mathcal{Z}$ for Problem \eqref{eq:stochastic_opt_pro} can be obtained through a chance-constrained problem. Thus, we can relax Problem \eqref{eq:stochastic_opt_pro} to the following chance-constrained problem.
\begin{theorem}[Inner chance-constrained approximation]\label{thm:chance_constrained_inner_approximation}
    Let $\epsilon \in (0, 1)$ be a given threshold and assume $h(x, z) \geq 0$ for all $(x, z) \in \mathbb{R}^n \times \mathbb{R}^d$. Define a uniform upper bound $M = \sup\limits_{x, z, \omega} g(x, z, \noise(\uncertaintyelement)) > 0$ on $g$ and let $\buffervar \geq \frac{\epsilon M}{1-\epsilon}$.
    Define the set \[
    E(x, z) = \left\{ \uncertaintyelement \in \uncertaintyspace : g(x, z, \noise(\uncertaintyelement)) + \buffervar \leq h(x, z)\right\},
    \]
    and consider the chance-constrained problem
        \begin{equation}
            \begin{aligned}
                \min_{z} & \quad s^\top z\\
                \subjectto & \quad \probmeas \left\{ E(x, z) \right\} \geq 1 - \epsilon, \text{ for all } x \in \safeset,
            \end{aligned}
            \label{eq:chance_constrainted_opt_pro}
        \end{equation}
   whose feasible set is given by $\mathcal{Z}' = \big\{z \in \mathbb{R}^d : \probmeas \left\{ E(x, z) \big\} \geq 1-\epsilon, \text{ for all } x \in \safeset \right\}$. Then we have that $\mathcal{Z}' \subseteq \mathcal{Z}$.
\end{theorem}
\begin{proof}
    Pick any $\bar{z} \in \mathcal{Z}'$. Our goal is to show that $\bar{z} \in \mathcal{Z}$. To this end, pick any $x\in \safeset$ and notice that 
    \[
        \begin{aligned}
        \mathbb{E}\left[g(x, \bar{z}, \noise(\uncertaintyelement)) \right] = &\int_{E(x,\bar{z})} g(x,\bar{z}, \noise(\uncertaintyelement)) \,d\probmeas(\uncertaintyelement) + \\
        &\int_{E(x,\bar{z})^c} g(x, \bar{z}, \noise(\uncertaintyelement)) \,d\probmeas(\uncertaintyelement).
        \end{aligned}
    \]
    Hence, we can derive the following
    \begin{equation}
        \begin{aligned}
        &\mathbb{E} \left\{ g(x, \bar{z}, \noise(\uncertaintyelement)) \right\}\\
        &\quad\; \leq (h(x, \bar{z}) - \buffervar)\mathbb{P}\{E(x,\bar{z})\} + M \mathbb{P}\{E(x,\bar{z})^c\} \\
        &\quad\; = h(x, \bar{z}) - \buffervar\mathbb{P}\{E(x,\bar{z})\} + M\mathbb{P}\{E(x,\bar{z})^c\}
        \end{aligned}
        \label{eq:proof_approx_SP_with_CCP}
    \end{equation}
    where the first inequality follows from the fact that $g(x, \bar{z},\noise(\uncertaintyelement))$ is less than or equal to $h(x, \bar{z}) - \buffervar$ on the set $E(x,\bar{z})$ and that $g$ is uniformly upper bounded by $M$ on the whole space $\uncertaintyspace$. 
    and the second inequality follows from $h(x, \bar{z}) \geq 0$ for all $(x, \bar{z}) \in \mathbb{R}^n \times \mathbb{R}^d$ and $\probmeas\{E(x, \bar{z})\} \in [0, 1]$.

    Now, by transitivity it holds that $\mathbb{E} \left\{ g(x, \bar{z}, \noise(\uncertaintyelement)) \right\} \leq h(x, \bar{z})$ if $h(x, \bar{z}) - \buffervar\mathbb{P}\{E(x,\bar{z})\} + M\mathbb{P}\{E(x,\bar{z})^c\} \leq h(x, \bar{z})$. Due to the feasibility of $\bar{z}$, it holds that $\mathbb{P}\{E(x,\bar{z})\} \geq 1 - \epsilon$ and $\mathbb{P}\{E(x,\bar{z})^c\} \leq \epsilon$. Furthermore, observe that $\buffervar \geq 0$ and $M > 0$, hence the second inequality holds if $- \buffervar\mathbb{P}\{E(x,\bar{z})\} + M\mathbb{P}\{E(x,\bar{z})^c\} \leq 0$. Restructuring this inequality, we arrive at $\buffervar \geq \frac{\epsilon M}{1-\epsilon}$, which is assumed to hold (by carefully choosing $\buffervar$).
    Therefore, we observe that $\bar{z} \in \mathcal{Z}$, thus concluding the proof of the theorem.
\end{proof}
What is left to show to conclude the proof is to show that $g(x, z, \noise(\uncertaintyelement)) = B(f(x) + \noise(\uncertaintyelement),\theta)$ and $h(x, z) = B(x) + \cmartingale$
satisfies the conditions of Theorem \ref{thm:chance_constrained_inner_approximation}, i.e., $h$ is non-negative and $g$ is uniformly bounded by $M$. Non-negative follows trivially by the definition of a \gls{sbf}, while boundedness of $B$ and consequently of $g$, can always be enforced.